\documentclass[conference,10pt,letterpaper]{IEEEtran}

\usepackage{amsmath,amssymb,amsthm}
\usepackage{xifthen}
\usepackage{mathtools}
\usepackage{enumerate}
\usepackage{microtype}
\usepackage{xspace}
\usepackage{bm}
\usepackage{fancyhdr}
\usepackage{lastpage}
\usepackage{bbm}

\newcommand{\mbs}[1]{\pmb{#1}}
\newcommand{\vect}[1]{{\lowercase{\mbs{#1}}}}
\newcommand{\mat}[1]{{\uppercase{\mbs{#1}}}}

\newcommand{\T}{{\scriptscriptstyle\mathsf{T}}}
\renewcommand{\H}{{\scriptscriptstyle\mathsf{H}}}

\renewcommand{\Re}[1][]{\ifthenelse{\isempty{#1}}{\operatorname{Re}}{\operatorname{Re}\left(#1\right)}}
\renewcommand{\Im}[1][]{\ifthenelse{\isempty{#1}}{\operatorname{Im}}{\operatorname{Im}\left(#1\right)}}

\newcommand{\SNR}{\mathsf{snr}}

\newcommand{\cv}{\vect{c}}

\newcommand{\rv}{\vect{r}}

\newcommand{\vv}{\vect{v}}

\newcommand{\xv}{\vect{x}}
\newcommand{\yv}{\vect{y}}

\newcommand{\Sigmam}{\pmb{\Sigma}}
\newcommand{\Gammam}{\pmb{\Gamma}}
\newcommand{\Lambdam}{\pmb{\Lambda}}

\newcommand{\Am}{\mat{a}}
\newcommand{\Bm}{\mat{b}}
\newcommand{\Cm}{\mat{c}}
\newcommand{\Dm}{\mat{d}}

\newcommand{\Mm}{\mat{M}}

\newcommand{\Qm}{\mat{q}}
\newcommand{\Rm}{\mat{r}}

\newcommand{\Um}{\mat{u}}
\newcommand{\Vm}{\mat{V}}
\newcommand{\Wm}{\mat{w}}
\newcommand{\Xm}{\mat{x}}
\newcommand{\Ym}{\mat{y}}

\newcommand{\Bc}{{\mathcal B}}
\newcommand{\Cc}{{\mathcal C}}

\newcommand{\Nc}{{\mathcal N}}
\newcommand{\Oc}{{\mathcal O}}

\newcommand{\Xc}{{\mathcal X}}

\newcommand{\CC}{\mathbb{C}}

\newcommand{\Id}{\mat{\mathrm{I}}}

\newcommand{\CN}[1][]{\ifthenelse{\isempty{#1}}{\mathcal{N}_{\mathbb{C}}}{\mathcal{N}_{\mathbb{C}}\left(#1\right)}}

\renewcommand{\P}[1][]{\ifthenelse{\isempty{#1}}{\mathbb{P}}{\mathbb{P}\left(#1\right)}}
\newcommand{\E}[1][]{\ifthenelse{\isempty{#1}}{\mathbb{E}}{\mathbb{E}\left[#1\right]}}
\newcommand{\Var}[1][]{\ifthenelse{\isempty{#1}}{\mathsf{Var}}{\mathsf{Var}\left[#1\right]}}
\newcommand{\I}[1][]{\ifthenelse{\isempty{#1}}{\mathbb{I}}{\mathbb{I}\left\{#1\right\}}}
\renewcommand{\det}[1][]{\ifthenelse{\isempty{#1}}{\mathrm{det}}{\mathrm{det}\left(#1\right)}}
\newcommand{\trace}[1][]{\ifthenelse{\isempty{#1}}{\mathrm{tr}}{\mathrm{tr}\left(#1\right)}}
\newcommand{\rank}[1][]{\ifthenelse{\isempty{#1}}{\mathrm{rank}}{\mathrm{rank}\left(#1\right)}}
\newcommand{\diag}[1][]{\ifthenelse{\isempty{#1}}{\mathrm{diag}}{\mathrm{diag}\left(#1\right)}}
\newcommand{\Cov}[1][]{\ifthenelse{\isempty{#1}}{\mathsf{Cov}}{\mathsf{Cov}\left(#1\right)}}

\newcommand{\defeq}{\triangleq}
\newcommand{\eqdef}{\triangleq}

\newtheorem{proposition}{Proposition}
\newtheorem{remark}{Remark}
\newtheorem{lemma}{Lemma}

\renewcommand{\rv}[1]{{\mathrm{#1}}}
\newcommand{\rvVec}[1]{\pmb{\mathrm{#1}}}
\newcommand{\rvMat}[1]{\pmb{\mathsf{#1}}}

\newcounter{enumi_saved}
\usepackage{answers}
\Newassociation{solution}{Solution}{solutionfile}

\AtBeginDocument{\Opensolutionfile{solutionfile}[\jobname]}
\AtEndDocument{\Closesolutionfile{solutionfile}\clearpage
}

\IfFileExists{MinionPro.sty}{
}{
}

\usepackage{fancyhdr}
\usepackage{lastpage}
\usepackage{bbm,empheq}
\usepackage{graphicx,subfigure}
\usepackage{setspace}
\usepackage{xcolor,enumitem}
\usepackage[hidelinks]{hyperref}

\usepackage{pgfplots}
\pgfplotsset{minor grid style={dotted,gray!25}}
\pgfplotsset{major grid style={dashed,gray!25}}

\usepackage[ruled,vlined,linesnumbered]{algorithm2e}

\usepackage{array,multirow}

\renewcommand{\rv}[1]{{\mathsf{#1}}}
\renewcommand{\rvVec}[1]{{\pmb{\mathsf{#1}}}}
\renewcommand{\rvMat}[1]{{\pmb{\mathsf{#1}}}}

\newcommand{\Span}[1][]{\ifthenelse{\isempty{#1}}{{\rm Span}}{{\rm Span}\left(#1\right)}}

\renewcommand{\SNR}{\mathsf{SNR}}

\renewcommand{\defeq}{:=}
\renewcommand{\eqdef}{=:}

\newcommand{\Xcal}{\mathcal{X}}
\newcommand{\LLR}{\rv{L}(\Xm\!\to\!{\Xm'})}
\newcommand{\meanLLR}{\E\big[\rv{L}(\Xm\!\to\!{\Xm'})\big]}
\newcommand{\varLLR}{\Var\big[\rv{L}(\Xm\!\to\!{\Xm'})\big]}

\title{Joint Constellation Design for the Two-User Non-Coherent Multiple-Access Channel} 

\author{\IEEEauthorblockN{Khac-Hoang Ngo\IEEEauthorrefmark{1}\IEEEauthorrefmark{2}, Sheng Yang\IEEEauthorrefmark{1}, Maxime Guillaud\IEEEauthorrefmark{2}, Alexis Decurninge\IEEEauthorrefmark{2}}
	\IEEEauthorblockA{\IEEEauthorrefmark{1}LSS, CentraleSup\'elec, 91190 Gif-sur-Yvette, France \\
		\IEEEauthorrefmark{2}Mathematical and Algorithmic Sciences Lab, Paris Research Center, Huawei Technologies, \\92100 Boulogne-Billancourt, France \\
		Email: {\{\href{mailto:ngo.khac.hoang@huawei.com}{ngo.khac.hoang}, \href{mailto:maxime.guillaud@huawei.com}{maxime.guillaud}, \href{mailto:alexis.decurninge@huawei.com}{alexis.decurninge}\}@huawei.com, \href{mailto:sheng.yang@centralesupelec.fr}{sheng.yang@centralesupelec.fr}}}
}

\mathtoolsset{showonlyrefs}

\newcommand{\sheng}[1]{{#1}}

\begin{document}
	
\maketitle
\date{\today}
\begin{abstract}
	We consider the joint constellation design problem for the two-user non-coherent multiple-access channel. Based on an analysis on the non-coherent maximum-likelihood~(ML) detection error, we propose novel design criteria so as to minimize the error probability. Based on these criteria, we propose a simple and efficient construction consisting in partitioning a single-user constellation. Numerical results show that our proposed metrics are meaningful, and can be used as objectives to generate constellations through numerical optimization that perform better than other schemes for the same transmission rate and power.
\end{abstract}


\section{Introduction} \label{sec:introduction}
In \sheng{multiple-input multiple-output}~(MIMO) communications, it is
usually assumed that the channel state information~(CSI) is known or
estimated (typically by sending pilots \sheng{and/or using feedback}),
and then used for \sheng{precoding} at the transmitter and detection at
the receiver. This is \sheng{known as} the {\em coherent} approach. On
the other hand, in the {\em non-coherent} approach, the transmission and
reception mechanisms are designed \sheng{without using the
\emph{a priori} knowledge of the
CSI}~\sheng{\cite{Hochwald2000unitaryspacetime,ZhengTse2002Grassman,Moser}}. 
In this paper, we consider the latter approach for the MIMO
multiple-access channel~(MAC) in block fading, \sheng{that is}, the channel remains unchanged during each coherence block of length~$T \ge 2$ and varies independently between blocks.


In the single-user case \sheng{with isotropic Rayleigh fading}, the optimal input signal at high signal-to-noise ratio~(SNR) is shown to be isotropic and uniformly distributed in the Grassmann manifold on $\CC^T$~\cite{Hochwald2000unitaryspacetime,ZhengTse2002Grassman}. Information is carried in the subspace of the transmitted signal matrix. The intuition behind this is that the signal subspace is not affected by the random fading coefficients. Motivated by this, there has been extensive research on the design of non-coherent constellations as a set of points on the Grassmann manifold. Many of these so-called Grassmannian constellation designs have been proposed, with a common criterion of maximizing the minimum pairwise chordal distance between the symbols~\cite{Gohary2009GrassmanianConstellations,Kammoun2007noncoherentCodes,Hoang2019TWC_Cubesplit}.

\sheng{In the multi-user case, a straightforward extension of the single-user
coherent approach is through the time division multiple access~(TDMA)
strategy, i.e., only one user is active at a time and transmits with an
optimal single-user constellation.}
\sheng{Another} straightforward extension is to \sheng{divide the coherent block into two parts:
1)~the training part in which orthogonal pilots are sent to estimate the
CSI for each user, and 2)~the data transmission part in which different
users communicate simultaneously}~\cite{Murugesan2007optimization}.
Although this approach achieves the optimal degree-of-freedom (DoF)
region in the two-user SIMO \sheng{MAC}~\cite{Hoang2018DoF_MAC}, its optimality in
terms of achievable rate and detection error \sheng{remains unclear}. 
\sheng{In the massive SIMO regime and when 
the channel changes rapidly}, some non-coherent modulation schemes have been proposed based on amplitude shift keying~\cite{Manolakos2016noncoh_Energybased_massiveSIMO}, or differential phase shift keying~\cite{Baeza2018noncoh_SIMO_MAC_DPSK_BICM}. 
In~\cite{HoangAsilomar2018multipleAccess}, a precoding-based
multiple-access scheme allowing efficient detection is proposed, but
offers no performance guarantee. 
To our knowledge, \sheng{there is no simple and effective joint
constellation construction criterion in the literature.}
%

In this work, we consider the two-user MIMO MAC and aim to design the
joint constellation and transmit powers so as to minimize the maximum
likelihood~(ML) symbol detection error. To this end, we analyze the
worst-case pairwise error probability and \sheng{derive} a design metric which
is the minimum \sheng{expected pairwise log-likelihood ratio~(PLLR)}
over the joint constellation. Furthermore, \sheng{from} the dominant
\sheng{term} of the expected PLLR at the high-SNR regime, we
\sheng{obtain a simplified metric that can be used for joint
constellation construction. Specifically, for any given pair of
constellation sizes, we can optimize the proposed metric over the set of
signal matrices.}
\sheng{In the case of symmetrical constellation size, we} propose a simple
construction consisting in partitioning a single-user constellation.
Numerical results show that our proposed metrics are meaningful \sheng{and
effective}, and the resulting constellations outperform a common
pilot-based scheme and the precoding design
in~\cite{HoangAsilomar2018multipleAccess}.

\sheng{In the following, we first present the system model and
formulate the problem. Then, we analyze the detection error probability
and derive the design metric. A new construction in the symmetrical case
is given, followed by some numerical results. The proofs can be found in the 
appendices}.


\textit{Notation:}
Random quantities are denoted with non-italic letters with sans-serif fonts, e.g., a scalar $\rv{x}$, a vector $\rvVec{v}$, and a matrix $\rvMat{M}$. 
Deterministic quantities are denoted 
with italic letters, e.g., a scalar $x$, a vector $\pmb{v}$, and a
matrix $\pmb{M}$.
The Euclidean norm is denoted by $\|\cdot\|$ and the Frobenius norm $\|\cdot\|_F$. The trace, transpose and conjugated transpose of $\pmb{M}$ are denoted $\trace[\Mm]$,
$\pmb{M}^\T$ and $\pmb{M}^\H$, respectively. 
$[n] \defeq \{1,2,\dots,n\}$. The logarithm $\log(\cdot)$
is to base $e$. The Grassmann manifold $G(\mathbb{C}^T,M)$ is defined as the space of $M$-dimensional subspaces in $\mathbb{C}^T$. 
In particular, $G(\mathbb{C}^T,1)$ is the Grassmannian of lines.

\section{System Model and Performance Metric} \label{sec:model}
We consider a MIMO MAC consisting of a receiver equipped with $N$ antennas and two users, user $k$ with $M_k$ antennas, $k = 1,2.$ The channel is assumed to be flat and block fading with equal-length and synchronous (across the users) coherence interval of length $T \ge 2$. That is, the channel matrix $\rvMat{H}_k \in \CC^{N\times M_k}$ of user~$k$ remains constant within each coherence block of $T$ channel uses and changes independently between blocks.
Furthermore, the {\em distribution} of $\rvMat{H}_k$ is assumed to be known to the receiver, but its {\em realizations} are unknown to both ends of the channel. 
We consider independent and identically distributed~(i.i.d) Rayleigh fading, namely, the rows of $\rvMat{H} \sheng{:=} [\rvMat{H}_1 \ \rvMat{H}_2]$ are independent and follow $\Cc\Nc(\mathbf{0},\Id_{M_1+M_2})$.

Within a coherence block \sheng{$t$,} $t = 1,2,\ldots$, each user $k$
sends a signal matrix\sheng{~symbol} $\rvMat{X}_k \in \CC^{T\times M_k}$, and the
receiver \sheng{observes} 
\begin{equation}
\rvMat{Y}[t] = \rvMat{X}_1[t] \rvMat{H}_1^\T[t] + \rvMat{X}_2[t] \rvMat{H}_2^\T[t] + \rvMat{Z}[t],
\label{eq:channel_model}
\end{equation}
where the additive noise $\rvMat{Z} \in \CC^{T\times N}$ has i.i.d.~$\Cc\Nc(0,1)$ entries independent of $\rvMat{H}_1$ and $\rvMat{H}_2$. 
We consider the power constraint
\sheng{$
	\tfrac{1}{n}\textstyle\sum_{t=1}^{n}\|\rvMat{X}_k[t]\|_F^2 \le P_k T, k =
	1,2,
	$}
where $n$ is the number of blocks spanned by a codeword.
Thus, $P_k$ is the SNR of the transmitted signal of user $k$ at each
receive antenna. For convenience, let us define $P:=\max\{ P_1, P_2
\}$.


We assume that the transmitted symbol $\rvMat{X}_k$ takes value from a {\em finite constellation} $\Xcal_k$ of fixed size $|\Xc_k| = 2^{R_k T}$, where $R_k$ (bits/channel use) is the transmission rate. 
Let us focus on one block and omit the
block index, and rewrite \eqref{eq:channel_model} as  
\begin{equation}
\rvMat{Y} = \rvMat{X} \rvMat{H}^\T + \rvMat{Z} = [
\rvMat{X}_1 \ \rvMat{X}_2] [
\rvMat{H}_1 \ \rvMat{H}_2]^\T + \rvMat{Z},
\end{equation}%
where the concatenated signal matrix $\rvMat{X} := [\rvMat{X}_1 \ \rvMat{X}_2]$ takes value from $\Xcal := \{ [
\Xm_1 \ \Xm_2]:\ \Xm_k\in\Xcal_k \}$. 
Our goal is to derive the
desirable properties of the set couple $(\Xcal_1, \Xcal_2)$ for a given rate
pair $(R_1,R_2)$ to achieve \sheng{low {\em symbol detection error
		probability}}. 
\begin{remark}
	\sheng{In the trivial case where one of the users, say user $2$, has zero
		rate~($R_2 = 0$), the joint constellation design problem boils down to the single-user
		constellation design. }
\end{remark}

The likelihood function $p_{\rvMat{Y} | \rvMat{X}}$ is given by
\begin{equation}
p_{\rvMat{Y} | \rvMat{X}}(\Ym | \Xm) = \frac{\exp(-\trace(\Ym^\H(\Id_T+\Xm\Xm^\H)^{-1}\Ym ))}{\pi^{NT}\det^N(\Id_T+\Xm\Xm^\H)}.
\label{eq:likelihood}
\end{equation}
Therefore, given the received \sheng{symbol} $\rvMat{Y} = \Ym$, the maximum likelihood~(ML) symbol detector is then 
\begin{align} \label{eq:MLdecoder}
\Xi(\Ym) 
&= \arg \max_{\Xm \in \Xc} 
\big(-\trace\big((\Id_T+\Xm\Xm^\H)^{-1}\Ym \Ym^\H \big) \big. \notag \\
&\qquad \qquad \big. - N\log\det(\Id_T+\Xm\Xm^\H)\big).
\end{align} 
We aim to design the constellations $\Xc_1$ and $\Xc_2$ so as to minimize the ML detection error $P_e(\Xcal) \!=\! \P[\Xi(\rvMat{Y}) \!\ne\! \rvMat{X}]$, i.e., 
\begin{equation}
\Xc^* = \arg\max_{\Xc} P_e(\Xcal). \label{eq:criterion_MLerror}
\end{equation}
\sheng{Since} the likelihood function depends on the symbol $\Xm\!$ only
through $\Xm \Xm^\H\!$, \sheng{the following proposition is
	straightforward.}
\begin{proposition}[Identifiability condition] \label{prop:identifiability}
	\sheng{With the ML detector}, the joint constellation $\Xc$ must satisfy $\Xm \Xm^\H \ne {\Xm'} {\Xm'}^\H$ 
	for any pair \sheng{of distinct symbols} $\Xm$ and ${\Xm'}$ in $\Xc$. 
\end{proposition}
\begin{remark} \label{rem:correlation}
	Although we do not consider correlated fading, we remark that if there is correlation between the antennas of the same user, namely, the rows of $\rvMat{H}$ are independent and follow $\Cc\Nc(\mathbf{0},\Rm)$ with $\Rm \defeq \big[\begin{smallmatrix}
	\Rm_1 & \mathbf{0} \\ \mathbf{0} & \Rm_2
	\end{smallmatrix}\big]$ where $\Rm_k$ is a $M_k \times M_k$ positive definite matrix, the solution to \eqref{eq:criterion_MLerror} can be expressed as $\bar{\Xc}_k \!=\! \{\Xm_k \Rm_k^{-1/2}\!:\! \Xm_k \!\in\! \Xc^*_k\}$ where $\Xc^*_k$, $k \!=\! 1,2$, are the solution to \eqref{eq:criterion_MLerror} for the considered uncorrelated fading but with a new power constraint $\frac{1}{n}\sum_{t=1}^{n}\|\rvMat{X}_k[t] {\Rm}_k^{-1/2}\|_F^2 \!\le\! P_k T,  k \!=\! 1,2.$
\end{remark}

In the next section, we derive more specific design criteria.

\section{Constellation Design Criteria} \label{sec:criteria}
With $\rvMat{X}$ uniformly distributed
in $\!\Xc$, $P_e(\Xcal)\!$ can be written as
\begin{equation} \label{eq:joint_SER}
P_e(\Xcal) = \frac{1}{|\Xcal|}\sum_{\Xm\in\Xcal}\mathbb{P}\left(
\Xi(\rvMat{Y}) \ne \Xm | \rvMat{X} = \Xm\right).
\end{equation}%
Denoting the pairwise error event $\{\Xm\!\to\!{\Xm'}\} \!\defeq\! 
\{ p_{\rvMat{Y} | \rvMat{X}}(\rvMat{Y} | \Xm) \le
 p_{\rvMat{Y} | \rvMat{X}}(\rvMat{Y} | {\Xm'}) | \rvMat{X} = \Xm\}$, 
we have the
following bounds on $P_e(\Xcal)$
\begin{multline}
\frac{1}{|\Xcal|} \max_{\Xm,{\Xm'}\in\Xcal: \Xm\ne{\Xm'}}
\mathbb{P}(\Xm\!\to\!{\Xm'})
\\ \le P_e(\Xcal) 
\le (|\Xcal|-1) \max_{\Xm,{\Xm'}\in\Xcal: \Xm\ne{\Xm'}}
\mathbb{P}(\Xm\!\to\!{\Xm'}). \label{eq:Pe_unionbound}
\end{multline}%
We see that for a given $|\Xcal|$, the symbol
detection error $P_e(\Xc)$ vanishes if and only if the \emph{worst-case
	pairwise error probability~(PEP)}, $\displaystyle\max_{\Xm,{\Xm'}\in\Xcal: \Xm\ne{\Xm'}}
\mathbb{P}(\Xm\!\to\!{\Xm'})$, vanishes. Therefore, {our goal from now on
is to minimize the worst-case PEP.} Let us rewrite the PEP as
$
\mathbb{P}(\Xm\!\to\!{\Xm'}) 
= \mathbb{P}\big(\rv{L}(\Xm\!\to\!{\Xm'}) \le 0 \big)
$ 
with the PLLR  $\rv{L}(\Xm\!\to\!{\Xm'}) \!\defeq\! \log
\frac{p_{\rvMat{Y} | \rvMat{X}}(\rvMat{Y} | \Xm)}{p_{\rvMat{Y} | \rvMat{X}}(\rvMat{Y} | {\Xm'})}$. Using \eqref{eq:likelihood}, we obtain 
\begin{align}
	&\rv{L}(\Xm\!\to\!{\Xm'})
= N\log\frac{\det\big( \Id_T+ {\Xm'}
	{\Xm'}^\H \big)}{\det\left( \Id_T+ \Xm \Xm^\H
	\right)}  \notag \\
&\qquad- \trace\Big( \big(( \Id_T\!+\! \Xm \Xm^\H )^{-1} - ( \Id_T\!+\! {\Xm'} {\Xm'}^\H )^{-1} \big) \rvMat{Y} \rvMat{Y}^\H \Big). \label{eq:LLR0}
\end{align}
Observe that $\LLR$ is a shifted weighted sum of Gamma random variables involving possibly both positive and negative weights. Therefore, $\mathbb{P}(\rv{L}(\Xm\!\to\!{\Xm'}) \le 0 )$ is hard to compute in general. 
We resort to the following bound on the PEP
\begin{align} \label{eq:Cantelli}
\!\!\!\!\!\!\!\!\!\P[\rv{L}(\Xm\!\to\!{\Xm'}) \!\le\! 0] 
\!\le\! 
\frac{\varLLR}{\varLLR \!+\! \meanLLR^2} 
\end{align}
which follows from Cantelli's inequality.
\sheng{Note that the upper bound decreases with $\meanLLR^2/\varLLR$.
We choose to relax the problem into maximizing the expected PLLR
$\meanLLR$.} 
\sheng{Although} maximizing $\!\meanLLR^2\!/\varLLR\!$ and maximizing
$\!\meanLLR\!$ are equivalent only when $\!\varLLR\!$ is constant over different symbol pairs, the
\sheng{relaxation} makes the problem tractable. 

\sheng{We further justify our choice by pointing out the 
connection of our problem to the following hypothesis testing problem. Let us consider
two hypotheses: ${H}_0:\ \{ \yv_i\}_{i=1}^N \sim
\Cc\Nc(\mathbf{0},\Id_T + \Xm \Xm^\H)$ and $H_1:\ \{
\yv_i\}_{i=1}^N \sim \Cc\Nc\big(\mathbf{0},\Id_T + {\Xm'}
{\Xm'}^\H\big)$} where $\{ \yv_i\}_{i=1}^N$ are realizations of $N$ columns of $\rvMat{Y}$. Then, the PEP $\P (\Xm\!\to\!{\Xm'})$ can be seen as the type-1 error probability of the optimal likelihood ratio test. From~\eqref{eq:Cantelli} and the fact that $\meanLLR^2/\varLLR \to \infty$ as $N\to \infty$, we have that $\P({\Xm}\to {\Xm'}) \to 0$ as $N\to \infty$ for any constellation satisfying the identifiability condition in Proposition~\ref{prop:identifiability}. (A proof is given in 
Appendix~\ref{proof:PEPvanishes}.) Switching the symbols' roles, we deduce that $\P({\Xm'}\to {\Xm}) \le \epsilon \in (0,1/2)$ for $N$ large enough. Then, from the Chernoff-Stein Lemma~\cite[Thm.11.8.3]{Cover2006elements}, we have that
\begin{align}
\lefteqn{
	\lim_{N\to \infty} \tfrac{1}{N} \log \P(\Xm\!\to\!{\Xm'})
} \notag \\
&= - D\big(\Cc\Nc(\mathbf{0},\Id_T + \Xm \Xm^\H) \ \| \  \Cc\Nc\big(\mathbf{0},\Id_T + {\Xm'} {\Xm'}^\H\big)\big) \\
&= - \meanLLR,
\end{align}
where $D(\cdot\|\cdot)$ denotes the Kullback-Leibler divergence.
Therefore, in the massive MIMO regime, maximizing $\meanLLR$ maximizes
the pairwise error exponent \sheng{with respect to the number of receive
antennas}. 

\sheng{Therefore}, we consider the following design criterion
\begin{equation}
\Xc^* = \arg\max_{\Xc} \min_{\Xm,{\Xm'}\in\Xcal: \Xm\ne{\Xm'}} \meanLLR. \label{eq:criterion_minMean}
\end{equation}
It follows from \eqref{eq:LLR0} and $\mathbb{E}[\rvMat{Y} \rvMat{Y}^\H] = N \big(\Id_T+ \Xm\Xm^\H)$ that
\begin{multline}
\meanLLR
= N\log\frac{\det\big( \Id_T+  {\Xm'}{\Xm'}^\H
	\big)}{\det\left( \Id_T+  \Xm \Xm^\H \right)} - N \notag \\
\quad+
N \trace\big( ( \Id
+ {\Xm'}{\Xm'}^\H )^{-1} \big) +
N\trace\big( ( \Id_T+ {\Xm'}{\Xm'}^\H )^{-1} \Xm
\Xm^\H \big). 
\end{multline}%
\begin{lemma} \label{lem:meanLLR_scaling}
	Let $\Xm$ and ${\Xm'}$ be such that~(s.t.) $\|\Xm\|_F^2 \!=\!
        \Theta(P)$ and $\|{\Xm'}\|_F^2 \!=\! \Theta(P)$, as
        $P\!\to\!\infty$. We have $\trace\big( ( \Id_T	\!+\! {\Xm'}{\Xm'}^\H )^{-1} \big) \!=\! O(1)$; $\log\frac{\det( \Id_T+ {\Xm'}{\Xm'}^\H
		)}{\det\left( \Id_T+ \Xm  \Xm^\H \right)} = O(1)$ if
                $\Span(\Xm) = \Span({\Xm'})$ and $\Theta(\log P)$ otherwise. Furthermore, $\trace\big( ( \Id_T+ {\Xm'}{\Xm'}^\H )^{-1} \Xm \Xm^\H \big) = O(1)$ if $\Span(\Xm) = \Span({\Xm'})$ and $\Theta(P)$ otherwise.
\end{lemma}
\begin{proof}
	Please see Appendix~\ref{proof:lem:meanLLR_scaling}.
\end{proof}
\sheng{We see that the only term that can scale up linearly  with $P$ is
$d(\Xm\!\to\!{\Xm'}) := \trace\big((\Id_T + {\Xm'} {\Xm'}^\H
)^{-1} \Xm \Xm^\H\big)$. Let  
$d_{\min}(\Xcal) \defeq \displaystyle\min_{\Xm,{\Xm'}\in\Xcal:
\Xm\ne{\Xm'}}d(\Xm\!\to\!{\Xm'})$, and we have the following design
criterion}
\begin{equation}
\Xc^* = \arg\max_{\Xc} d_{\min}(\Xcal). \label{eq:criterion_minTrace}
\end{equation} 

\subsection{The Single-User Case} \label{sec:single_user}
In the single-user case with $M$ transmit antenna, it is known that the
high-SNR optimal input signal belongs to the Grassmann
manifold~\cite{ZhengTse2002Grassman}. We consider Grassmannian
constellation~\cite{Gohary2009GrassmanianConstellations} $\Xc \!\subset\!
G(\CC^T,M)$, thus ${\Xm}^\H{\Xm} \!=\! \frac{PT}{M}\Id_M, \forall \Xm \!\in\! \Xc$. It follows that
\begin{equation}
d(\Xm\!\to\!{\Xm'}) 
= PT\bigg(1 - \alpha_{P,T,M}\frac{\| {\Xm'}^\H \Xm \|_F^2}{(PT)^2}\bigg) , 
\end{equation}%
where 
$
\alpha_{P,T,M} := \left(\tfrac{1}{PT}+\frac{1}{M} \right)^{-1}. 
$
Therefore, the design criterion~\eqref{eq:criterion_minTrace} is equivalent to
$
\Xc 
=
\arg\displaystyle\min_{\Xcal} \max_{\Xm,{\Xm'}\in\Xcal:\Xm\ne{\Xm'}}\! \| {\Xm'}^\H \Xm \|_F^2.
$
This coincides with the common criterion of maximizing the minimum pairwise chordal distance~\cite{Conway1996packing,Gohary2009GrassmanianConstellations,Kammoun2007noncoherentCodes,Hoang2019TWC_Cubesplit}.

\subsection{The Two-User Case}
\label{sec:alternating}
In the two-user case, we assume for simplicity that $M_1=M_2=M$,
although the general case follows straightforwardly. We first develop
\begin{align}
d(\Xm\to{\Xm'})
&= \trace\big( \Xm_1^\H (\Id_T+ {\Xm'}{\Xm'}^\H )^{-1} \Xm_1 \big) \notag \\ 
&\quad + \trace\big( \Xm_2^\H (\Id_T + {\Xm'}{\Xm'}^\H )^{-1} \Xm_2
\big), 
\end{align}
where we recall that $\Xm := [\Xm_1 \ \Xm_2]$, ${\Xm'} := [\Xm'_1 \ \Xm'_2]$ with $\Xm_k,\Xm'_k\in\Xcal_k$, $k\in\{1,2\}$, and ${\Xm'}\ne\Xm$. 
\sheng{There are two types of error event.}
\begin{enumerate} [leftmargin=*]
  \item \sheng{Simultaneous detection error},
	i.e., $\Xm_1\ne\Xm'_1$, $\Xm_2 \ne \Xm'_2$:
	\begin{align}
	d(\Xm\!\to\!{\Xm'}) &= \trace\big( \Xm_1^\H ( \Id_T\!+\! \Xm'_1 {\Xm'}_1^\H \!+\! \Xm'_2 {\Xm'}_2^\H )^{-1} \Xm_1
	\big) \notag \\
	&\quad+ \trace\big( \Xm_2^\H ( \Id_T\!+\! {\Xm'}_1^\H
	\Xm'_1 \!+\! {\Xm'}_2^\H \Xm'_2 )^{-1} \Xm_2 \big). 
	\end{align}%
	\item \sheng{One sided detection error}, i.e.,
          $\Xm_k=\Xm'_k$, $\Xm_{l}\ne\Xm'_{l}$, $k\ne
          l\in\left\{ 1,2 \right\}$:
	\begin{align}
	d(\Xm\!\to\!{\Xm'}) &= 
        {\trace\big( \Xm_k^\H ( \Id_T\!+\! {\Xm}_k {\Xm}_k^\H \!+\! \Xm'_l {\Xm'}_l^\H )^{-1} \Xm_k
		\big)}
                \notag \\
	&\quad+ \trace\big( \Xm_l^\H ( \Id_T+ {\Xm'}_k^\H
	\Xm'_k + {\Xm'}_l^\H \Xm'_l )^{-1} \Xm_l \Big). 
	\end{align}%
	\end{enumerate}
\sheng{Let us define}
\begin{align}
\!\!\!\!\!d_{12}(\Xcal) &:=\!\!  \min_{\Xm_1\ne\Xm'_1\in\Xcal_1\atop
{\Xm}_2\in\Xcal_2}\!\!\! \trace\big(\Xm_1^\H ( \Id_T \!+\! \Xm'_1 {\Xm'}_1^\H \!+\! {\Xm}_2 {\Xm}_2^\H )^{-1} \Xm_1
\big), \label{eq:d12} \\
\!\!\!\!\!d_{21}(\Xcal) &:=\!\!  \min_{\Xm_2\ne\Xm'_2\in\Xcal_2\atop
\Xm'_1\in\Xcal_1}\!\!\! \trace\big(\Xm_2^\H ( \Id_T \!+\! {\Xm}_1 {\Xm}_1^\H \!+\! \Xm'_2 {\Xm'}_2^\H )^{-1} \Xm_2
\big). \label{eq:d21}
\end{align}%
\sheng{Since $\trace\big( \Xm_k^\H ( \Id_T\!+\! {\Xm}_k {\Xm}_k^\H \!+\! \Xm'_l {\Xm'}_l^\H )^{-1} \Xm_k
		\big) \le M$, $\forall\,k\ne l$, considering both types
                of error, we obtain}
\begin{align}
\min\left\{ d_{12}(\Xcal),
d_{21}(\Xcal) \right\} &\le d_{\min}(\Xcal) \notag\\&\le  \min\left\{
d_{12}(\Xcal), d_{21}(\Xcal) \right\} + M. 
\end{align}%
Therefore, $d_{\min}(\Xcal)$ is within a constant gap to $\min\left\{
d_{12}(\Xcal), d_{21}(\Xcal) \right\}$, and $d_{\min}(\Xcal)$ scales linearly with $P$ when $P$ is large if and only if $\min\left\{
d_{12}(\Xcal), d_{21}(\Xcal) \right\}$ does so. Based on this observation, we propose the
following design criterion
\begin{align} \label{eq:criterion_minAltTrace}
\Xcal^* = \arg\max_{\Xcal} \ \min\left\{ d_{12}(\Xcal), d_{21}(\Xcal) \right\}. 
\end{align}%
\sheng{Note that the above metric, though simplified, is not convex and
is hard to optimize. 
Nevertheless, since} each of the traces in $d_{12}(\Xcal)$ and
$d_{21}(\Xcal)$ involves a symbol pair in the constellation of one user
and a symbol from the other user's constellation, we propose a method to approximately solve \eqref{eq:criterion_minAltTrace} as follows. First $\Xc_1$ and $\Xc_2$ are initialized. Then, we iteratively alternate between fixing $\Xc_2$, optimizing $\Xc_1$ by $\Xcal_1 = \arg\displaystyle\max_{\Xcal_1} d_{12}(\Xcal)$, and fixing $\Xc_1$, optimizing $\Xc_2$ by $\Xcal_2 = \arg\displaystyle\max_{\Xcal_2} d_{21}(\Xcal)$. We refer to this as {\em alternating optimization}.
%
It has fewer variables to optimize than directly solving~\eqref{eq:criterion_minMean}, \eqref{eq:criterion_minTrace}, or \eqref{eq:criterion_minAltTrace}.

\section{\sheng{The Symmetrical Rate Case} } \label{sec:partitioning}
\sheng{In the following, we focus on the case with symmetrical user
rate. In addition, we let the average power of both constellations
to be the same to further simplify the
optimization.}\footnote{\sheng{Note that using maximum power for both
users may be suboptimal for non-coherent MAC.}}
\sheng{To further reduce the solution space, we make the 
(suboptimal)~assumption that the individual constellations are from the
Grassmann manifold.}
From \sheng{the practical perspective}, this is desirable
since the constellation is oblivious to the presence of the other user. 

\sheng{Nevertheless}, there must be constraints
	between the symbols of different users. For instance, if the
	constellations are such that $\Xm_1=\Xm_2$ can occur, then $d_{12}(\Xcal)$
	and $d_{21}(\Xcal)$ are upper-bounded by a constant. 
This can be developed in a formal
way as follows. 
An upper bound can be obtained by removing either term inside the
inverse, namely,
\begin{align}
d_{12}(\Xcal) &\le \min\Big\{ \min_{\Xm_1\ne\Xm'_1\in\Xcal_1} \trace\big(\Xm_1 ^\H( \Id_T+ \Xm'_1{\Xm'}_1^\H)^{-1} \Xm_1 \big),\Big.\notag \\
&~\Big. \min_{\Xm_1\in\Xcal_1, {\Xm}_2\in\Xcal_2} \trace\big(\Xm_1^\H (\Id_T+ {\Xm}_2 {\Xm}_2^\H )^{-1} \Xm_1 \big) \Big\}, \label{eq:d12_ub} \\
d_{21}(\Xcal) &\le \min\Big\{ \min_{\Xm_2\ne\Xm'_2\in\Xcal_2}
\trace\big(\Xm_2^\H ( \Id_T+ \Xm'_2 {\Xm'}_2^\H)^{-1} \Xm_2
\big), \Big.\notag \\ 
&~\Big. \min_{\Xm_1\in\Xcal_1, {\Xm}_2\in\Xcal_2}
\trace\big(\Xm_2^\H ( \Id_T+ {\Xm}_1 {\Xm}_1^\H)^{-1} \Xm_2\big)
\Big\}. \label{eq:d21_ub} 
\end{align}%
Therefore, for $d_{12}(\Xcal)$ and $d_{21}(\Xcal)$ to be large,
the upper bounds must
be large. The next proposition makes the argument precise. 
\begin{proposition}[Necessary condition] \label{prop:necessary}
	Let $\Xcal_1$ and $\Xcal_2$ be s.t. $\Xm_k^\H\Xm_k = \frac{PT}{M} \Id_M$,
	$\forall\,\Xm_k\in\Xcal_k$, $k\in\{1,2\}$. If the following lower bound on the
	$d$-values holds for some $c\in[0,1/M]$
	\begin{equation}
	\min\left\{ d_{12}(\Xcal), d_{21}(\Xcal) \right\} \ge 
	PT\left(1 - \alpha_{P,T,M} \, c \right), \label{eq:necc0}
	\end{equation}%
	where $\alpha_{P,T,M} := \bigl(\frac{1}{PT}+\frac{1}{M} \bigr)^{-1}$, then we must have
	\begin{align}
	&\frac{1}{(PT)^2}\max\Big\{ \max_{\Xm_1\ne\Xm'_{1}\in\Xcal_1} \big\|
	{\Xm'}_1^\H \Xm_1
	\big\|_F^2,
	\max_{\Xm_2\ne\Xm'_{2}\in\Xcal_2} \big\|
	{\Xm'}_2^\H\Xm_2
	\big\|_F^2, \Big.\notag \\
	&\qquad\qquad \Big.
	\max_{\Xm_1\in\Xcal_1, \Xm_2\in\Xcal_2} \|
	{\Xm}_2^\H\Xm_1 \|_F^2
	\Big\} \le c.
	\label{eq:necc}
	\end{align}%
\end{proposition}
\begin{proof}
	The proof follows the same steps as in the single-user case, applying
	on the upper bounds \eqref{eq:d12_ub} and \eqref{eq:d21_ub}.  
\end{proof}
The above shows that symbol pairs from
different users should fulfill similar distance criteria as symbol pairs
from the same constellation when it comes to identifiability conditions. However, it is unclear whether
\eqref{eq:necc} alone is enough to guarantee a large $d$-value. 
In the following, we shall show that these conditions are indeed sufficient if $c$ is small.  

\begin{proposition}[Sufficient condition] \label{prop:sufficient}
	Let $\Xcal_1$ and $\Xcal_2$ be s.t. $\Xm_k^\H\Xm_k = \frac{PT}{M} \Id_M$,
	$\forall\,\Xm_k\in\Xcal_k$, $k\in\{1,2\}$. 
	If
	\begin{align}
	&\frac{1}{(PT)^2}\max\Big\{ \max_{\Xm_1\ne\Xm'_{1}\in\Xcal_1} \big\|
	{\Xm'}_1^\H \Xm_1
	\big\|_F^2,
	\max_{\Xm_2\ne\Xm'_{2}\in\Xcal_2} \big\|
	{\Xm'}_2^\H\Xm_2
	\big\|_F^2, \Big.\notag \\
	&\qquad \qquad \Big.
	\max_{\Xm_1\in\Xcal_1, \Xm_2\in\Xcal_2} \|
	{\Xm}_2^\H\Xm_1 \|_F^2
	\Big\} \le c,
	\label{eq:suff}
	\end{align}
	for some $c\in[0,1/M]$, then we have 
	\begin{equation}
	\!\min\left\{ d_{12}(\Xcal), d_{21}(\Xcal) \right\} \!\ge\! 
	PT\Big(1 - 2 \big( \tfrac{1}{PT}+\tfrac{1}{M}-\sqrt{c} \big)^{-1}
	c \Big).\! \label{eq:suff2}
	\end{equation}%
\end{proposition}
\begin{proof}
	Please see Appendix~\ref{proof:sufficient}.
\end{proof}
The two propositions above motivates the following simplified design criterion
\begin{multline} \label{eq:criterion_minChordal}
\Xc^* = \arg\min_{\Xcal} \ 
\max\Big\{ \max_{\Xm_1\ne\Xm'_{1}\in\Xcal_1} \big\|
{\Xm'}_1^\H \Xm_1
\big\|_F^2, \Big. \\
\qquad\quad\Big.
\max_{\Xm_2\ne\Xm'_{2}\in\Xcal_2} \big\|
{\Xm'}_2^\H\Xm_2
\big\|_F^2, \max_{\Xm_1\in\Xcal_1, \Xm_2\in\Xcal_2} \big\|
{\Xm}_2^\H\Xm_1 \big\|_F^2
\Big\}.
\end{multline}
Based on this criterion, we propose a construction as follows. Let $\Xcal_{\text{SU}}$ be a single-user constellation satisfying 
$
\frac{1}{(PT)^2}\displaystyle\max_{\Xm \ne {\Xm'} \in \Xc_{\text{SU}}} \|
{\Xm'}^\H \Xm \|_F^2 \le c,
$
for some $c\in\bigl[0,\frac{1}{M}\bigr]$. 
%
%
%
%
Then, we can generate $\Xc_1$ and $\Xc_2$ by partitioning  $\Xcal_{\text{SU}}$ into two disjoint subsets. In this way, from Proposition~\ref{prop:sufficient}, we can
guarantee that 
\begin{equation}
d_{\min}(\Xcal_1,\Xcal_2) \ge 
PT\left(1 - 2 \big( \tfrac{1}{PT}+\tfrac{1}{M}-\sqrt{c} \big)^{-1} \,
c \right). \label{eq:tmp833}
\end{equation}%
In other words, with such a construction, the joint constellation design
problem becomes essentially an individual constellation design problem. A random partition would be sufficient to guarantee \eqref{eq:tmp833}, although one can smartly partition the set
$\Xcal_{\text{SU}}$ to improve over \eqref{eq:tmp833}. The optimal partition problem is equivalent to a min-max graph bipartitioning problem. Also note that for the right-hand side of \eqref{eq:tmp833} to scale linearly with $P$, we must have that $c < \big[\big(\frac{1}{2PT} + \frac{1}{2M} +\frac{1}{16}\big)^{1/2}-\frac{1}{4}\big]^2$. This limits the size of the initial single-user constellation $\Xcal_{\text{SU}}$. 

\section{Numerical Results} \label{sec:performance}
In this section, we consider the single transmit antenna case ($M_1
\!=\! M_2 \!=\! 1$) and focus on the \sheng{symmetrical rate} setting
$R_1 \!=\! R_2 \!=\! B/T$ \sheng{with equal power of both users}.  We
assume Grassmannian signaling, i.e., each constellation symbol is
\sheng{a} unit vector representative of a point on $G(\CC^T,1)$ scaled with $\sqrt{PT}$. 
We solve numerically \eqref{eq:criterion_minMean}, \eqref{eq:criterion_minTrace}, and the alternating optimization. In general, we want to solve the optimization on the manifold 
$
\displaystyle\max_{\Xc = \Xc_1 \times \Xc_2} \displaystyle\min_{\Xm \ne {\Xm'} \in \Xc} f(\Xm, {\Xm'}), 
$
where $\Xc_k \subset G(\CC^T,1), k \in \{1,2\}$,
and 
$f(\Xm, {\Xm'})$ is customized according to the considered criterion.
For smoothness, we use the approximation $\max_i x_i \approx \epsilon \log \sum_i \exp (x_i/\epsilon)$ with a small $\epsilon$ and obtain 
$
\min_{\Xc  = \Xc_1 \times \Xc_2} \epsilon \log \sum_{\Xm \ne {\Xm'} \in \Xc} \exp \big(-\frac{f(\Xm, {\Xm'})}{\epsilon}\big).
$
In Appendix~\ref{app:gradient}, we compute the gradient of this new objective function. 
Finally, we resort to the manopt toolbox~\cite{manopt} to solve the optimization by gradient descent on the manifold for a fixed SNR of $P = 30$~dB (even when the performance of the resulting constellations is benchmarked at other SNR values). 

We compare the performance of the constellations optimized with the proposed criteria and the precoding design in \cite{HoangAsilomar2018multipleAccess} (with Precoder Type II therein), and a coherent pilot-based scheme with orthogonal pilots and quadrature amplitude modulation~(QAM) data symbols. In the pilot-based scheme, the receiver uses ML decoder or a systematic decoder consisting of minimum-mean-square-error~(MMSE) channel estimation, MMSE equalization, and QAM demapper. 



In Fig.~\ref{fig:joint_SER_T5B5N4}, we plot the joint symbol error
rate~(SER) \eqref{eq:joint_SER} of these schemes for $T = 5$, $B = 5$,
and $N = 4$. We observe that the constellations optimized with the
metrics $\min_{\Xm \ne {\Xm'} \in \Xc} \meanLLR$
\eqref{eq:criterion_minMean} and $d_{\min}(\Xc)$
\eqref{eq:criterion_minTrace} achieve similar performance and are the
best among the schemes with \sheng{the same rate pair}. The performance of the alternatively optimized constellation is slightly inferior to that and performs better the pilot-based scheme at $\SNR < 18$~dB. The partitioning design (with random partition) and the precoding design~\cite{HoangAsilomar2018multipleAccess} have similar performance.
\begin{figure}[h!]
	\centering
%
%
\begin{tikzpicture}[scale=.69,style={mark size=3pt,line width=3pt}]

\begin{axis}[%
width=4.2in,
height=2.95in,
at={(0.758in,0.481in)},
scale only axis,
xmin=4,
xmax=18,
xtick={ 2, 4, 6, 8, 10, 12, 14, 16, 18},
xlabel style={font=\color{white!15!black}},
xlabel={SNR (dB)},
ymode=log,
ymin=1e-06,
ymax=1,
yminorticks=true,
ylabel style={font=\color{white!15!black}},
ylabel={Joint Symbol Error Rate},
axis background/.style={fill=white},
xmajorgrids,
ymajorgrids,
yminorgrids,
legend style={at={(0.02,0.015)}, anchor=south west, legend cell align=left, align=left, draw=white!15!black, nodes={scale=0.92}}
]
\addplot [color=blue, dashed, line width=1pt, mark=o, mark options={solid, blue}]
  table[row sep=crcr]{%
2	0.2611\\
4	0.1161\\
6	0.0442\\
8	0.0192\\
10	0.00684\\
12	0.00305\\
14  8.9e-4  \\
16  3.2e-4  \\
18  9e-5  \\
};
\addlegendentry{Precoding~\cite{HoangAsilomar2018multipleAccess}}

\addplot [color=blue, dashed, line width=1pt, mark=square, mark size=2.5pt, mark options={solid, blue}]
  table[row sep=crcr]{%
2	0.241\\
4	0.1037\\
6	0.0388\\
8	0.0149\\
10	0.00583\\
12	0.00229\\
14  1.01e-3  \\
16  4.6333e-04  \\
18  1.35e-4  \\
};
\addlegendentry{Partitioning (Sec.~\ref{sec:partitioning})}

\addplot [color=green, dashdotted, line width=1pt, mark size=2.5pt, mark=*, mark options={solid, green}]
  table[row sep=crcr]{%
2	0.2315\\
4	0.1007\\
6	0.0355\\
8	0.0088\\
10	0.00186\\
12	0.00038\\
14  9e-5  \\
16  1.7e-5  \\
18  2e-06  \\
};
\addlegendentry{Maximizing $\min_{\Xm \ne {\Xm'} \in \Xc} \E\big[\rv{L}(\Xm \!\rightarrow\! {\Xm'})\big]$~\eqref{eq:criterion_minMean}}

\addplot [color=red, dashdotted, line width=1pt, mark=triangle, mark size=4pt, mark options={solid, rotate=180, red}]
  table[row sep=crcr]{%
2	0.2359\\
4	0.1005\\
6	0.0335\\
8	0.0077\\
10	0.00189\\
12	0.00039\\
14  1e-4  \\
16  2e-5  \\
18  2.0930e-06  \\
};
\addlegendentry{Maximizing $d_{\min}(\Xc)$~\eqref{eq:criterion_minTrace}}

\addplot [color=red, dashdotted, line width=1pt, mark=triangle, mark size=4pt, mark options={solid, rotate=270, red}]
  table[row sep=crcr]{%
2	0.2427\\
4	0.099\\
6	0.0348\\
8	0.0114\\
10	0.00239\\
12	0.00062\\
14  1.5000e-04  \\
16  2.6667e-05  \\
18  4.3396e-06  \\
};
\addlegendentry{Alternating optimization (Sec.~\ref{sec:alternating})}

\addplot [color=black, dotted, line width=1pt, mark size=3pt, mark=x, mark options={solid, black}]
  table[row sep=crcr]{%
2	0.3649\\
4	0.1701\\
6	0.0622\\
8	0.0206\\
10	0.00535\\
12	0.00113\\
14  2.75e-4  \\
16  4.6667e-05  \\
18  5.2830e-06  \\
};
\addlegendentry{Pilot-based scheme, ML decoder}

\addplot [color=black, dotted, line width=1pt, mark size=3pt, mark=+, mark options={solid, black}]
  table[row sep=crcr]{%
2	0.4912\\
4	0.2957\\
6	0.1561\\
8	0.0648\\
10	0.02655\\
12	0.0089\\
14  0.0029  \\
16  8.1333e-04  \\
18  2.070e-04  \\
};
\addlegendentry{Pilot-based scheme, systematic decoder}


\end{axis}

\end{tikzpicture}%
	\caption{The joint SER of the proposed constellations \sheng{compared
        to the precoding
        design~\cite{HoangAsilomar2018multipleAccess} and a pilot-based
        scheme} for $T = 5$, $B = 5$, and $N = 4$.}
	\label{fig:joint_SER_T5B5N4}
\end{figure}
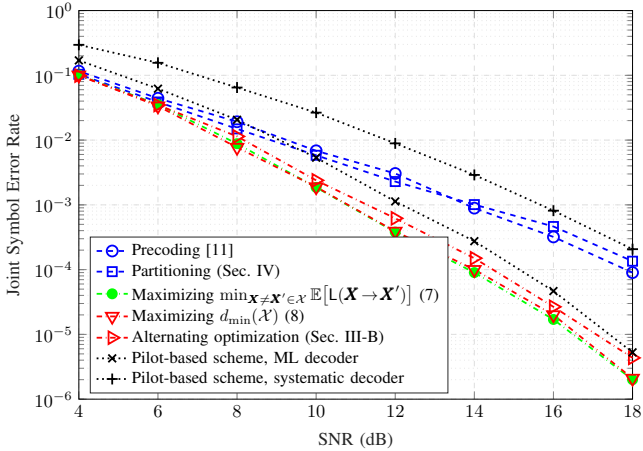

In the same setting, we show the values of the metrics $\min_{\Xm \ne {\Xm'} \in \Xc} \meanLLR$ and $d_{\min}(\Xc)$ for these schemes in Fig.~\ref{fig:metric_T5B5N4}. As can be seen, $d_{\min}(\Xc)$ is very close to $\min_{\Xm \ne {\Xm'} \in \Xc} \meanLLR$ for $\SNR \ge 20$~dB. The schemes with low ML SER in Fig.~\ref{fig:joint_SER_T5B5N4} has a large value of these metrics. This confirms that our proposed metrics are meaningful for constellation design and evaluation.
\begin{figure}[h!]
	\centering
%
%
\begin{tikzpicture}[scale=.69,style={mark size=3pt,line width=3pt}]
\begin{axis}[%
width=4.3in,
height=2.95in,
at={(0.758in,0.481in)},
scale only axis,
xmin=4,
xmax=32,
xtick={ 4,  8, 12, 16, 20, 24, 28, 32},
xlabel style={font=\color{white!15!black}},
xlabel={SNR (dB)},
ymode=log,
ymin=1,
ymax=4000,
yminorticks=true,
ylabel style={font=\color{white!15!black}},
ylabel={$\frac{1}{N}\displaystyle\min_{\Xm \ne {\Xm'} \in \Xc} \E\big[\rv{L}(\Xm\!\to\!{\Xm'})\big]$ (lines) and $d_{\min}(\Xc)$ (markers)},
axis background/.style={fill=white},
xmajorgrids,
ymajorgrids,
yminorgrids,
legend style={at={(0.02,0.6)}, anchor=south west, legend cell align=left, align=left, draw=white!15!black, nodes={scale=0.98}}
]

\addplot [color=blue, dashed, line width=1.2pt, mark=o, mark options={solid, blue}]
table[row sep=crcr]{%
	0 0 \\
	2 0.1\\
};
\addlegendentry{Precoding~\cite{HoangAsilomar2018multipleAccess}}
\addplot [color=blue, dashed, line width=1.2pt, mark=square, mark size=2.5pt, mark options={solid, blue}]
table[row sep=crcr]{%
	0 0 \\
	2 0.1 \\
};
\addlegendentry{Partitioning (Sec.~\ref{sec:partitioning})}
\addplot [color=green, dashed, line width=1.2pt, mark size=2.5pt, mark=*, mark size=1.5pt, mark options={solid, green}]
table[row sep=crcr]{%
	0 0\\
	2 0.1\\
};
\addlegendentry{Maximizing $\min_{\Xm \ne {\Xm'} \in \Xc} \E\big[\rv{L}(\Xm \!\rightarrow\! {\Xm'})\big]$~\eqref{eq:criterion_minMean}}
\addplot [color=red, dashdotted, line width=1.2pt, mark=triangle, mark size=5pt, mark options={solid, rotate=180, red}]
table[row sep=crcr]{%
	0 0 \\
	2 0.1 \\
};
\addlegendentry{Maximizing $d_{\min}(\Xc)$~\eqref{eq:criterion_minTrace}}
\addplot [color=red, dashdotted, line width=1.2pt, mark=triangle, mark size=5pt, mark options={solid, rotate=270, red}]
table[row sep=crcr]{%
	0 0\\
	2 0.1\\
};
\addlegendentry{Alternating optimization (Sec.~\ref{sec:alternating})}

\addplot [color=black, dotted, line width=1.2pt, mark=x, mark size=3pt, mark options={solid, black}]
table[row sep=crcr]{%
	0 0 \\
	2 0.1\\
};
\addlegendentry{Pilot-based scheme}

\addplot [color=blue, draw=none, line width=1.2pt, mark=o, mark options={solid, blue}, only marks]
  table[row sep=crcr]{%
4	2.0984174692586\\
8	2.76832883966429\\
12	3.99725392682579\\
16	6.43251485758493\\
20	12.4964679468903\\
24	27.7068969429122\\
28	65.905134379964\\
32	161.851322268047\\
};

\addplot [color=blue, draw=none, line width=1.2pt, mark=square, mark size=2.5pt, mark options={solid, blue}, only marks]
  table[row sep=crcr]{%
4	2.4530278069444\\
8	2.80249388991699\\
12	3.35118811217385\\
16	4.48755906325657\\
20	6.95623160312973\\
24	13.1229503464044\\
28	28.5992541600736\\
32	67.4684587542839\\
};

\addplot [color=green, draw=none, line width=1.2pt, mark size=2.5pt, mark=*, mark size=1.5pt, mark options={solid, green}, only marks]
  table[row sep=crcr]{%
4	4.34220546560225\\
8	8.76266395151446\\
12	19.7187663040643\\
16	47.0706581851467\\
20	115.712897339682\\
24	287.950843057077\\
28	720.372079359971\\
32	1806.56265768064\\
};

\addplot [color=red, draw=none, line width=1pt, mark=triangle, mark size=5pt, mark options={solid, rotate=180, red}, only marks]
  table[row sep=crcr]{%
4	4.3464688329927\\
8	8.77558365828902\\
12	19.6958021896097\\
16	47.0056108432417\\
20	115.353932467924\\
24	286.662328462442\\
28	716.961264890367\\
32	1797.82000894721\\
};

\addplot [color=red, draw=none, line width=1pt, mark=triangle, mark size=5pt, mark options={solid, rotate=270, red}, only marks]
  table[row sep=crcr]{%
4	2.96500933983253\\
8	4.72402248839876\\
12	8.90136927256601\\
16	19.2869967842949\\
20	45.3297042173332\\
24	110.727845140037\\
28	274.993253736676\\
32	687.606391975158\\
};

\addplot [color=black, draw=none, line width=1.2pt, mark=x, mark size=3pt, mark options={solid, black}, only marks]
  table[row sep=crcr]{%
4	4.16043243774584\\
8	8.07240926770741\\
12	17.7412600554605\\
16	41.9607293631702\\
20	102.76954607794\\
24	255.503207836167\\
28	639.148347728078\\
32	1602.81958502378\\
};

\addplot [color=blue, dashed, line width=1.2pt, forget plot]
  table[row sep=crcr]{%
4	4.6831571071635\\
8	5.11382138872007\\
12	6.0950695476042\\
16	8.51421853078472\\
20	14.5716695582491\\
24	29.7794958900382\\
28	67.9766949306034\\
32	163.9224690669\\
};
\addplot [color=blue, dashed, line width=1.2pt, forget plot]
  table[row sep=crcr]{%
4	4.7949975453669\\
8	5.0648025726405\\
12	5.54776081552482\\
16	6.54379105763996\\
20	9.00424949078443\\
24	15.1676726320667\\
28	30.6426603604889\\
32	69.5113403621576\\
};
\addplot [color=green, dashed, line width=1.2pt, forget plot]
  table[row sep=crcr]{%
4	6.83426351762013\\
8	11.1712732390129\\
12	22.0472282459524\\
16	49.3578508780232\\
20	117.955393578539\\
24	290.237730706755\\
28	722.658227347212\\
32	1808.84851103076\\
};
\addplot [color=red, dashdotted, line width=1.2pt, forget plot]
  table[row sep=crcr]{%
4	6.82702099573026\\
8	11.1508452106548\\
12	21.9933802602065\\
16	49.21992009009\\
20	117.592404709795\\
24	288.940204506734\\
28	719.238225116876\\
32	1800.09660424661\\
};
\addplot [color=red, dashdotted, line width=1.2pt, forget plot]
  table[row sep=crcr]{%
4	5.39860470837305\\
8	7.08825100019371\\
12	11.2361209999254\\
16	21.6097008449763\\
20	47.6475613807262\\
24	113.043764628921\\
28	277.30840053726\\
32	689.921230958891\\
};
\addplot [color=black, dotted, line width=1.2pt, forget plot]
  table[row sep=crcr]{%
4	6.4366874224319\\
8	10.2748314397425\\
12	19.9125462841818\\
16	44.1193270822303\\
20	104.923044845152\\
24	257.654669063748\\
28	641.298996585805\\
32	1604.96991027907\\
};
\end{axis}
\end{tikzpicture}%
	\caption{The value of the metrics $\frac{1}{N}\min_{\Xm \ne {\Xm'} \in \Xc} \E\big[\rv{L}(\Xm\!\to\!{\Xm'})\big]$ (lines) and $d_{\min}(\Xc)$ (markers) for the considered schemes for $T \!=\! 5$ and $B \!=\! 5$.}
	\label{fig:metric_T5B5N4}
\end{figure}

In Fig.~\ref{fig:joint_SER_B8N4}, we consider larger constellations ($B = 8$) for which numerical solutions to \eqref{eq:criterion_minMean}, \eqref{eq:criterion_minTrace}, and the alternating optimization become cumbersome. However, the random partitioning construction, which is based on our metrics, achieves good performance and outperforms the pilot-based scheme.
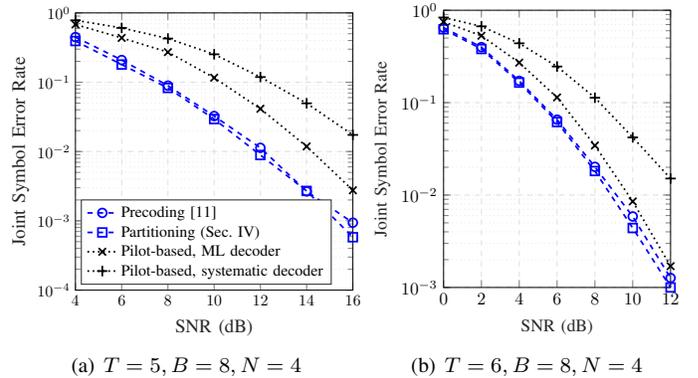
\begin{figure}[h!]
	\centering
	\hspace{-.2cm}
	\subfigure[$T=5, B = 8, N = 4$]{
%
%
\begin{tikzpicture}[scale=.66,style={mark size=3pt,line width=3pt}]

\begin{axis}[%
width=2.2in,
height=2.2in,
at={(0.758in,0.481in)},
scale only axis,
xmin=4,
xmax=16,
xtick={4, 6, 8, 10, 12, 14, 16},
xlabel style={font=\color{white!15!black}},
xlabel={\large SNR (dB)},
ymode=log,
ymin=0.0001,
ymax=1,
yminorticks=true,
ylabel style={font=\color{white!15!black}},
ylabel={\large Joint Symbol Error Rate},
axis background/.style={fill=white},
xmajorgrids,
ymajorgrids,
yminorgrids,
legend style={at={(0.02,0.02)}, anchor=south west, legend cell align=left, align=left, draw=white!15!black}, nodes={scale=.9}
]
\addplot [color=blue, dashed, line width=1pt, mark=o, mark size=2.5pt, mark options={solid, blue}]
  table[row sep=crcr]{%
0	0.843\\
2	0.653\\
4	0.448\\
6	0.209\\
8	0.0892\\
10	0.0325\\
12	0.0113\\
14	0.002660\\
16	0.00094\\
};
\addlegendentry{Precoding~\cite{HoangAsilomar2018multipleAccess}}

\addplot [color=blue, dashed, line width=1pt, mark=square, mark size=2.5pt, mark options={solid, blue}]
  table[row sep=crcr]{%
0	0.814\\
2	0.624\\
4	0.393\\
6	0.18\\
8	0.0827\\
10	0.0294\\
12	0.0089\\
14	0.00270\\
16	0.00058\\
};
\addlegendentry{Partitioning (Sec.~\ref{sec:partitioning})}

\addplot [color=black, dotted, line width=1pt, mark=x, mark size=3pt, mark options={solid, black}]
  table[row sep=crcr]{%
0	0.949\\
2	0.857\\
4	0.681\\
6	0.436\\
8	0.2718\\
10	0.116\\
12	0.0414\\
14	0.01184\\
16	0.00276\\
};
\addlegendentry{Pilot-based, ML decoder}

\addplot [color=black, dotted, line width=1pt, mark=+, mark size=3pt, mark options={solid, black}]
  table[row sep=crcr]{%
0	0.959\\
2	0.897\\
4	0.782\\
6	0.608\\
8	0.4269\\
10	0.2529\\
12	0.1187\\
14	0.04966\\
16	0.01736\\
};
\addlegendentry{Pilot-based, systematic decoder}

\end{axis}
\end{tikzpicture}%
}
	\hspace{-.2cm}
	\subfigure[$T=6, B = 8, N = 4$]{
%
%
\begin{tikzpicture}[scale=.66,style={mark size=3pt,line width=3pt}]

\begin{axis}[%
width=1.8in,
height=2.2in,
at={(0.758in,0.481in)},
scale only axis,
xmin=0,
xmax=12,
xtick={0, 2, 4, 6, 8, 10, 12},
xlabel style={font=\color{white!15!black}},
xlabel={SNR (dB)},
ymode=log,
ymin=0.001,
ymax=1,
yminorticks=true,
ylabel style={font=\color{white!15!black}},
ylabel={Joint Symbol Error Rate},
axis background/.style={fill=white},
xmajorgrids,
ymajorgrids,
yminorgrids,
legend style={legend cell align=left, align=left, draw=white!15!black}
]
\addplot [color=blue, dashed, line width=1pt, mark=o, mark size=2.5pt, mark options={solid, blue}]
table[row sep=crcr]{%
0	0.649\\
2	0.399\\
4	0.172\\
6	0.0655\\
8	0.0202\\
10	0.005866666666667\\
12	0.001266666666667\\
14	0\\
16	0\\
};

\addplot [color=blue, dashed, line width=1pt, mark=square, mark size=2.5pt, mark options={solid, blue}]
table[row sep=crcr]{%
0	0.623\\
2	0.381\\
4	0.165\\
6	0.0617\\
8	0.0182\\
10	0.0044\\
12	0.001\\
14	0\\
16	0\\
};

\addplot [color=black, dotted, line width=1pt, mark=x, mark size=3pt, mark options={solid, black}]
table[row sep=crcr]{%
0	0.749\\
2	0.53\\
4	0.271\\
6	0.1139\\
8	0.0344\\
10	0.008533333333333\\
12	0.0017\\
14	0\\
16	0\\
};

\addplot [color=black, dotted, line width=1pt, mark=+, mark size=3pt, mark options={solid, black}]
table[row sep=crcr]{%
0	0.835\\
2	0.67\\
4	0.439\\
6	0.2461\\
8	0.1131\\
10  0.042066666666667\\
12	0.0151\\
14	0\\
16	0\\
};

\end{axis}
\end{tikzpicture}%
}
	\hspace{-.2cm}

	\caption{The joint SER of the partitioning design in comparison with the precoding design and a pilot-based scheme for $T \!\in\! \{5,6\}$, $B \!=\! 8$, and $N \!=\! 4$.}
	\label{fig:joint_SER_B8N4}
\end{figure}

\section{Conclusion} \label{sec:conclusion}
\sheng{This work is our first attempt of joint constellation design for
non-coherent MIMO MAC. We have derived some closed-form metrics which
turned out to be effective for such purpose. A next step is to
investigate the asymmetrical rate case in which power optimization also
plays a key role. 
}

\bibliographystyle{IEEEtran}
\bibliography{IEEEabrv,./biblio}

\appendices

\section{Proof that $\lim\limits_{N\to\infty}\P({\Xm}\to {\Xm'}) = 0$ for any pair of distinct symbols $\Xm$ and $\Xm'$ of an identifiable joint constellation} \label{proof:PEPvanishes}
Recall that $\P({\Xm}\to {\Xm'}) = \P(\LLR \le 0)$. Let $\rvMat{Y}_0 \defeq (\Id_T+\Xm\Xm^\H)^{-\frac12}\rvMat{Y}$ be a ``whitened'' version of $\rvMat{Y}$, then $\rvMat{Y}_0$ is a Gaussian matrix with independent columns following $\Cc\Nc(\mathbf{0},\Id_T)$. Let
\begin{align}
\Lambdam \defeq (\Id_T+\Xm\Xm^\H)^{\frac12}(\Id_T+\Xm'{\Xm'}^\H)^{-1} (\Id_T+\Xm\Xm^\H)^{\frac12} - \Id_T.
\end{align}
From \eqref{eq:LLR0}, we develop $\LLR$ as
\begin{align}
\LLR &= -N\log\det(\Id_T+ \Lambdam) + \trace(\Lambdam \rvMat{Y}_0 \rvMat{Y}_0^\H) \\
&= -N\log\det(\Id_T+ \Lambdam) + \sum_{i=1}^{T} \lambda_i \rv{g}_i,
\end{align}
where $\lambda_1,\dots,\lambda_T$ are $T$ eigenvalues of $\Lambdam$, and $\rv{g}_1,\dots,\rv{g}_T$ are independent Gamma random variables with shape $N$ and scale 1. It follows that
\begin{align}
\meanLLR &= -N\log\det(\Id_T+ \Lambdam) + \sum_{i=1}^{T} \lambda_i \\
&= -N\log\det(\Id_T+ \Lambdam) + N\trace(\Lambdam), \\
\varLLR &= N \sum_{i=1}^{T} \lambda^2_i = N\trace[\Lambdam^2].
\end{align}
For any joint constellation satisfying the identifiability condition in Proposition~\ref{prop:identifiability}, we have $\Xm\Xm^\H \ne \ \Xm'{\Xm'}^\H$, thus $\Lambdam \ne \mathbf{0}$. Therefore, $\trace(\Lambdam) - \log\det(\Id_T+ \Lambdam)$ is strictly larger than $0$. We have that
\begin{align}
\frac{\meanLLR^2}{\varLLR} &= N \frac{\big(\trace(\Lambdam) - \log\det(\Id_T+ \Lambdam)\big)^2}{\trace[\Lambdam^2]} \\
&\to \infty, \quad \text{as $N\to \infty$}.
\end{align}
From this and \eqref{eq:Cantelli}, we conclude that $\lim\limits_{N\to\infty}\P({\Xm}\to {\Xm'}) = 0$ for any pair of distinct symbols $\Xm$ and $\Xm'$ of a joint constellation satisfying the identifiability condition.

\section{Proof of Lemma~\ref{lem:meanLLR_scaling}} \label{proof:lem:meanLLR_scaling}
$\trace\left( \big( \Id_T	+ {\Xm'}{\Xm'}^\H \big)^{-1} \right) = O(1)$ is straightforwardly because the eigenvalues of $( \Id_T	+ {\Xm'}{\Xm'}^\H )^{-1}$ are all smaller than $1$. 

The input matrix $\Xm$ can be decomposed into an orthonormal matrix $\Wm \in \CC^{T\times (M_1+M_2)}$ whose columns span the column space of $\Xm$ and a full-rank spanning matrix $\Dm$. That is
$
\Xm = \Wm\Dm^\H, 
$
where $\|\Dm\|^2_F = \Theta(P)$.
Similarly, 
$
{\Xm'} = {\Wm'}{\Dm'}^\H,
$
for some orthonormal matrix ${\Wm'} \in \CC^{T\times (M_1+M_2)}$ and some full-rank  spanning matrices ${\Dm'}$ s.t. $\|{\Dm'}\|^2_F = \Theta(P)$. We assume without loss of generality~(w.l.o.g.) that the column subspaces of $\Xm$ and ${\Xm'}$ share $r\le M_1 + M_2$ eigenmodes and thus express $\Wm$ and ${\Wm'}$ as
\begin{align} \label{eq:tmp996}
\Wm = [\Um \ \Vm], \quad \text{and} \quad {\Wm'} = [\Um \ {\Vm'}],
\end{align}
with $\Um \in \CC^{T \times r}$, $\Vm \in \CC^{T \times (M_1+M_2-r)}$, and ${\Vm'} \in \CC^{T \times (M_1+M_2-r)}$ s.t. $\Um^\H\Um = \Id$, $\Vm^\H\Vm = \Id$, ${\Vm'}^\H{\Vm'} = \Id$, $\Um^\H\Vm = \mathbf{0}$, $\Um^\H{\Vm'} = \mathbf{0}$, and $\Vm^\H{\Vm'} = \mathbf{0}$.
In the following, $\sigma_i(\Mm)$ denotes the $i$-th eigenvalue of a matrix $\Mm$ in decreasing order.

\subsubsection{Proof that $\log\frac{\det\big( \Id_T+ {\Xm'}{\Xm'}^\H \big)}{\det\left( \Id_T+ \Xm  \Xm^\H \right)} = O(1)$ if $\Span(\Xm) = \Span({\Xm'})$ and $\Theta(\log P)$ otherwise}
The following lemma is useful for our proof.
\begin{lemma} \label{lemma:hermitian_pertubation}
	Consider $T\times T$ Hermitian matrices $\Am$ and $\Bm$ whose entries are functions of a parameter $P$. Assuming that $\|\Bm\|_F^2 = O(1)$ as $P\to \infty$, then 
	\begin{align} \label{eq:tmp1016}
	\sigma_i(\Am + \Bm) = \sigma_i(\Am) + O(1), \quad \forall i \in [T], P \to \infty.
	\end{align}
\end{lemma}
\begin{proof}
	From the Hoffman-Wielandt Theorem~\cite[Them.6.3.5]{Horn2012MatrixAnalysis}, we have that
	\begin{align} \label{eq:tmp1022}
	\sum_{i=1}^{T} (\sigma_i(\Cm) - \sigma_i(\Dm))^2 \le \|\Cm-\Dm\|_F^2,
	\end{align}
	for $T \times T$ Hermitian matrices $\Cm$ and $\Dm$. Then, \eqref{eq:tmp1016} follows by applying \eqref{eq:tmp1022} with $\Cm = \Am + \Bm$ and $\Dm = \Bm$.
\end{proof}
Let 
\begin{align}
\Gammam &= (\Id_T+{\Xm'}{\Xm'}^\H)(\Id_T+\Xm \Xm^\H)^{-1} \\
&= {\Xm'}{\Xm'}^\H(\Id_T+\Xm \Xm^\H)^{-1} + (\Id_T+\Xm \Xm^\H)^{-1}.
\end{align} 
Applying Lemma~\ref{lemma:hermitian_pertubation} with $\Am =  {\Xm'}{\Xm'}^\H(\Id_T+\Xm \Xm^\H)^{-1}$ and $\Bm = (\Id_T+\Xm \Xm^\H)^{-1}$, we have that
\begin{align}
&\!\!\!\!\!\!\!\sigma_i(\Gammam) = \sigma_i\big({\Xm'} {\Xm'}^\H(\Id_T+{\Xm}{\Xm}^\H)^{-1}\big) + O(1) \\
&\!\!\!\!\!\!\!=\!\begin{cases}
\sigma_i\big( {\Xm'}^\H(\Id_T\!+\!{\Xm}{\Xm}^\H)^{-1} {\Xm'}\big) \!+\! O(1), &i \le M_1 \!+\! M_2, \\
O(1), &i > M_1\!+\!M_2.
\end{cases} \label{eq:tmp929}
\end{align}
Recalling the decomposition $\Xm = \Wm \Dm^\H$, we have that ${\Xm}{\Xm}^\H = {\Wm} \Sigmam {\Wm}^\H$ with $\Sigmam \defeq \Dm^\H\Dm$. Let ${\Wm}_\perp$ be the orthonormal complement of ${\Wm}$, i.e., $[{\Wm} \ {\Wm}_\perp]$ is an unitary matrix. We can write that $\Id_T + {\Xm}{\Xm}^\H = {\Wm} (\Id_T +  {\Sigmam}) {\Wm}^\H + {\Wm}_\perp {\Wm}_\perp^\H$, and $(\Id_T+ {\Xm}{\Xm}^\H)^{-1} = {\Wm} \big(\Id_T +  {\Sigmam}\big)^{-1} {\Wm}^\H + {\Wm}_\perp {\Wm}_\perp^\H$. We can expand
\begin{align}
{\Xm'}^\H(\Id_T\!+\!{\Xm}{\Xm})^{-1} {\Xm'}^\H &= {\Xm'}^\H {\Wm} (\Id_T \!+\! {\Sigmam})^{-1} {\Wm}^\H {\Xm'} \notag \\
&\quad + {\Xm'}^\H {\Wm}_\perp {\Wm}_\perp^\H {\Xm'}. \label{eq:tmp1043}
\end{align}
Recalling that ${\Xm'} = {\Wm'} {\Dm'}^\H$ and using \eqref{eq:tmp996}, we have that 
\begin{align}
&{\Xm'}^\H {\Wm} (\Id_T +  {\Sigmam})^{-1} {\Wm}^\H \Xm' \notag \\
&\quad=  {\Dm'}^\H
\begin{bmatrix}
\Id_r & \mathbf{0} \\ \mathbf{0} & \mathbf{0}
\end{bmatrix}
(\Id_T + {\Sigmam})^{-1}
\begin{bmatrix}
\Id_r & \mathbf{0} \\ \mathbf{0} & \mathbf{0}
\end{bmatrix}
{\Dm'}\\
&\quad= {\Dm'}_1^\H  (\Id_r + {\Sigmam}_1)^{-1} \Dm'_1, 
\end{align}
where $\Dm'_1$ contains the first $r$ columns of ${\Dm'}$ and ${\Sigmam}_1$ denotes the top-left $r\times r$ block of ${\Sigmam}$, respectively. Therefore,
\begin{align}
\|{\Xm'}^\H {\Wm} (\Id_T + {\Sigmam})^{-1} {\Wm}^\H {\Xm'}\|^2_F &\le \frac{\|\Dm'_1\|_F^2}{1+{\sigma}_{\min}(\Sigmam_1)} = O(1),
\end{align}
where ${\sigma}_{\min}(\Sigmam_1)$ is the smallest eigenvalue of ${\Sigmam}_1$. 
With this, we apply Lemma~\ref{lemma:hermitian_pertubation} to \eqref{eq:tmp1043} and obtain $\sigma_i\big( {\Xm'}^\H(\Id_T+{\Xm}{\Xm}^\H)^{-1}{\Xm'}\big) = \sigma_i\big({\Xm'}^\H {\Wm}_\perp {\Wm}_\perp^\H {\Xm'} \big) + O(1)$. Plugging this in \eqref{eq:tmp929}, we get that
\begin{align}
\sigma_i(\Gammam) &=\begin{cases}
\sigma_i\big({\Xm'}^\H {\Wm}_\perp {\Wm}_\perp^\H {\Xm'} \big) + O(1), &\quad i \le M_1 + M_2, \\
O(1), &\quad i > M_1+M_2,
\end{cases} \\
&=
\begin{cases}
\Theta(P), &\quad i \le M_1 + M_2-r, \\
O(1), &\quad i > M_1+M_2-r.
\end{cases} 
\end{align}

If $\Span[\Xm] = \Span[{\Xm'}]$, we have that $r = M_1 + M_2$, thus $\sigma_i(\Gammam) = O(1)$ for all $i \in [T]$. Thus $\log\det(\Gammam) = O(1)$. Otherwise, $\sigma_i(\Gammam) = \Theta(P)$ for some $i$ and we have that $\log\det[\Id_T + \Gammam] = \sum_{i=1}^T \log \sigma_i = \Theta(\log P)$.

\subsubsection{Proof that $\trace\big( ( \Id_T+ {\Xm'}{\Xm'}^\H )^{-1} \Xm \Xm^\H \big) = O(1)$ if $\Span(\Xm) = \Span({\Xm'})$ and $\Theta(P)$ otherwise}
We expand
\begin{align}
&\Xm^\H(\Id+{\Xm'}{\Xm'}^\H)^{-1} \Xm \notag \\
&= \Xm^\H\left(\Id-{\Xm'}(\Id+{\Xm'}^\H{\Xm'})^{-1}{\Xm'}^\H\right) \Xm \\
&= \Dm \Wm^\H\Big(  \Id_T\!-\!{\Wm'}{\Dm'}^\H (\Id_T\!+\! {\Dm'} {\Wm'}^\H {\Wm'}{\Dm'}^\H)^{-1} {\Dm'}{\Wm'}^\H \Big) \Wm\Dm^\H \\
&= \Dm\Dm^\H - \Dm 
\bigg[\begin{matrix}
\Id_r & \mathbf{0} \\
\mathbf{0} & \mathbf{0}
\end{matrix}\bigg]
{\Dm'}^\H \big(\Id_T \!+\! {\Dm'} {\Dm'}^\H \big)^{-1} {\Dm'}
\bigg[\begin{matrix}
\Id_r & \mathbf{0} \\
\mathbf{0} & \mathbf{0}
\end{matrix}\bigg]
\Dm^\H \\
&= \Dm\Dm^\H - \begin{bmatrix}
P^2T^2 \Dm_1
{\Dm'}_1^\H \big(\Id_r + \Dm'_1 {\Dm'}_1^\H \big)^{-1} {\Dm'}_1
\Dm_1^\H 
& \mathbf{0} \\
\mathbf{0} & \mathbf{0}
\end{bmatrix}
\end{align}
where $\Dm_1$ and $\Dm'_1$ contain the first $r$ columns of $\Dm$ and ${\Dm'}$, respectively.
Thus,
\begin{align}
\lefteqn{\trace[(\Id+{\Xm'}{\Xm'}^\H)^{-1} \Xm \Xm^\H]} \notag \\ 
&= \trace[{\Dm}_2 {\Dm}_2^\H] \notag \\
&\quad+ \trace[\Dm_1\Dm_1^\H - \Dm_1
{\Dm'}_1^\H \big(\Id_r + \Dm'_1 {\Dm'}_1^\H \big)^{-1} \Dm'_1
\Dm_1^\H] \\
&= \|{\Dm}_2\|_F^2 + \trace[ \big(\Id_r + \Dm'_1 {\Dm'}_1^\H\big)^{-1} \Dm_1^\H\Dm_1],
\end{align}
where $\Dm_2$ contains the last $M_1+M_2-r$ columns of $\Dm$. Since $\big(\Id_r + \Dm'_1 {\Dm'}_1^\H\big)^{-1} \preceq (1+\sigma_{\min}({\Dm'}_1^\H \Dm'_1))^{-1}\Id$ where $\sigma_{\min}({\Dm'}_1^\H \Dm'_1)$ is the smallest eigenvalue of ${\Dm'}_1^\H \Dm'_1$, we have that 
\begin{align}
\|{\Dm}_2\|_F^2 &\le \trace[(\Id+{\Xm'}{\Xm'}^\H)^{-1} \Xm \Xm^\H] \\
&\le \|{\Dm}_2\|_F^2 + \frac{\|{\Dm}_1\|_F^2}{1+\sigma_{\min}({\Dm'}_1^\H \Dm'_1)}.
\end{align}
If $\Span[\Xm] = \Span\big({\Xm'}\big)$, we have  $r=M_1 + M_2$ and thus $\Dm_2$ is an empty matrix. Therefore, 
\begin{equation}
\trace[(\Id\!+\!{\Xm'}{\Xm'}^\H)^{-1} \Xm \Xm^\H] \le \frac{\|{\Dm}_1\|_F^2}{1+\sigma_{\min}({\Dm'}_1^\H \Dm'_1)} = O(1).
\end{equation}
Otherwise, $r<M_1+M_2$ and $\|{\Dm}_2\|_F^2 = \Theta(P)$, thus $\trace[(\Id+{\Xm'}{\Xm'}^\H)^{-1} \Xm \Xm^\H] = \Theta(P)$. 

\section{Proof of Proposition~\ref{prop:sufficient}} \label{proof:sufficient}
Due to the symmetry, it is enough to focus on $d_{12}(\Xcal)$. 
Let us rewrite $\Xm'_1
{\Xm'}_1^\H + {\Xm}_2 {\Xm}_2^\H = \Xm_{12} \Xm_{12}^\H$ where
$\Xm_{12} := \left[\Xm'_1 \ \Xm_2\right] \in \Xcal$. Then, the trace in \eqref{eq:d12}
becomes
\begin{align}
&\trace\Big(\Xm_1^\H \big(\Id_T+ {\Xm}_{12} {\Xm}_{12}^\H \big)^{-1} \Xm_1
	\Big)\notag\\ 
&= \trace\left(\Xm_1^\H \Xm_1\right) - \trace \Big(\Xm_1^\H{\Xm}_{12} \big(\Id_T+ {\Xm}_{12}^\H {\Xm}_{12} \big)^{-1} {\Xm}_{12}^\H \Xm_1 \Big)  
\\
&= PT- \trace \big(\Xm_1^\H \Um
\Sigmam \big(\Id_T+ \Sigmam^2 \big)^{-1} \Sigmam \Um^\H
\Xm_1  \big), 
\end{align}%
where $\Xm_{12} = \Um \Sigmam \Vm^\H$ with $\Um\in\mathbb{C}^{r\times T}$, $\Vm\in\mathbb{C}^{2M\times	r}$ being orthogonal matrices, and
$r$ being the rank of $\Xm_{12}$; $\Sigmam$ is a diagonal matrix with
the $r$ singular values of $\Xm_{12}$ in decreasing order. 
Then, since $\big(\Id_T+ \Sigmam^2 \big)^{-1} \preceq
(1+\sigma^2_{\min}(\Xm_{12}))^{-1} \Id$ with $\sigma^2_{\min}(\Xm_{12})$ being the
minimum non-zero singular value of $\Xm_{12}$, we have
\begin{align}
&\!\!\!\!\!\!\!\!\!\trace\Big(\Xm_1^\H \big(\Id_T+ {\Xm}_{12} {\Xm}_{12}^\H \big)^{-1} \Xm_1
	\Big)\notag\\ 
&\!\!\!\!\!\!\!\!\!\ge PT - \big(1+\sigma^2_{\min}(\Xm_{12})\big)^{-1} \trace \left( \Xm_1^\H \Um \Sigmam  \Sigmam \Um^\H \Xm_1  \right) \\
&\!\!\!\!\!\!\!\!\!= PT- \big(1+\sigma^2_{\min}(\Xm_{12})\big)^{-1} \left\| \Xm_{12}^\H
\Xm_1 \right\|_F^2 \\
&\!\!\!\!\!\!\!\!\!= PT - \big(1\!+\!\sigma^2_{\min}(\Xm_{12})\big)^{-1}
\big(\| {\Xm'}_1^\H \Xm_1 \|_F^2 \!+\!
\| \Xm_2^\H \Xm_1 \|_F^2 \big). \label{eq:tmp922} 
\end{align}%

  From
  \eqref{eq:tmp922}, the key is to find a lower bound on the non-zero singular value
  $\sigma_{\min}(\Xm_{12})$. The following lemma is useful for that
  purpose.
\begin{lemma} \label{lem:min_singular_value}
  Let $\Qm := \left[\begin{smallmatrix} \Id_m & \Am_{m\times n} \\
    \Am_{m\times n}^\H &
    \Id_n \end{smallmatrix}\right]$
    be positive semidefinite. Then, the $m+n$ eigenvalues of $\Qm$ are 
    \begin{multline}
      1+\sigma_1(\Am), \ldots, 1+\sigma_{\min\{m,n\}}(\Am), 1, \ldots, 1, \\ 
      1-\sigma_{\min\{m,n\}}(\Am), \ldots, 1-\sigma_1(\Am). 
    \end{multline}%
\end{lemma}
\begin{proof}
  The singular value decomposition of $\Am$ leads to a block
  diagonalization of $\Qm$ with $2\times 2$ blocks. The result then
  follows immediately. 
\end{proof}
Applying Lemma~\ref{lem:min_singular_value} to the matrix $\Qm = \frac{M}{PT} \Xm_{12}^\H
\Xm_{12}$ with $\Am = \frac{M}{PT} {\Xm'}_1^\H \Xm_2$, we see that
the minimum non-zero eigenvalues of $\Qm$ is $1-\sigma_k(\frac{M}{PT}
{\Xm'}_1^\H \Xm_2)$ if there exists at least one singular value of $\frac{M}{PT}
{\Xm'}_1^\H \Xm_2$
strictly smaller than $1$ and $\sigma_k(\frac{M}{PT}
{\Xm'}_1^\H \Xm_2)$ is the largest among such
values. Otherwise, if all singular values of $\frac{M}{PT} {\Xm'}_1^\H
\Xm_2$ are $1$, the minimum non-zero eigenvalue of $\Qm$ is two. In
any case, the minimum non-zero eigenvalue of $\Qm$ is lower
bounded
by 
$
1-\|
\frac{M}{PT} \Xm'_1 \Xm_2^\H \| \ge 1-M\sqrt{c}. 
$
Hence, $\sigma^2_{\min}(\Xm_{12}) \ge PT\big(\frac{1}{M} - \sqrt{c}\big)$. Plugging this into \eqref{eq:tmp922} yields \eqref{eq:suff2}.

\section{The Constellation Numerical Optimization} \label{app:gradient}
Recall the (approximate) constellation optimization
\begin{align}
\min_{\Xc  = \Xc_1 \times \Xc_2} \epsilon \log \sum_{\Xm \ne {\Xm'} \in \Xc} \exp \Big(-\frac{f(\Xm, {\Xm'})}{\epsilon}\Big), \label{eq:tmp1102}
\end{align}
where 
$
\Xc_k = \{\sqrt{PT}\cv_{ki}\}_{i=1}^{|\Xc_k|} \subset G(\CC^T,1)$ with $\cv_{k1}, \cv_{k2}, \dots, \cv_{k|\Xc_k|}$ being unit-norm vectors, $k \in \{1,2\}$, 
a joint symbol $\Xm \in \Xc$ is formed as $\Xm = [\xv_1 \ \xv_2]$ for $\xv_1 \in \Xc_1$ and $\xv_2 \in \Xc_2$, 
and 
$f(\Xm, {\Xm'})$ is customized according to the considered criterion. This smooth optimization is, however, jointly over multiple points on the Grassmannian of lines. To tackle this, we construct the matrix $\Cm \defeq [\cv_{11} \dots \cv_{1|\Xc_1|} \ \cv_{21} \dots \cv_{2|\Xc_2|}] \in \CC^{T\times (|\Xc_1|+|\Xc_2|)}$, then $\Cm$ belongs to the oblique manifold  $\Oc\Bc(T,|\Xc_1|+|\Xc_2|)$ defined as 
\begin{multline}
\Oc\Bc(n,m) \defeq \\\big\{ \Mm = [\vv_1 \dots \vv_m]\in \CC^{n\times m}: \|\vv_1\| = \dots = \|\vv_m\| = 1\big\}.
\end{multline}
The oblique manifold $\Oc\Bc(n,m)$ can be seen as an embedded Riemannian manifold of $\CC^{n\times m}$
endowed with the usual inner product, or as the product manifold of $m$ unit spheres in $\CC^T$. Then, the optimization problem \eqref{eq:tmp1102} can be reformulated as a single-variable optimization on this oblique manifold as
\begin{equation}
\min_{\Cm \in \Oc\Bc(T,|\Xc_1|+|\Xc_2|)} \ \underbrace{\epsilon \log\sum_{\Xm = \sqrt{PT}[\cv_{1i}\ \cv_{2l}] \atop \ne {\Xm'} = \sqrt{PT}[\cv_{1j}\ \cv_{2m}]} \exp \Big(-\frac{f(\Xm, {\Xm'})}{\epsilon}\Big)}_{\large \eqdef g(\Cm)}. \label{eq:opt_oblique}
\end{equation}
To solve this, we need to compute the Riemannian gradient of the function $g(\Cm)$ on the manifold. According to~\cite[Sec.3.6]{AbsMahSep2008optManifolds}, 
the Riemannian gradient can be computed by projection as
\begin{align}
\nabla_R g(\Cm) = (\Id_T - \Cm\Cm^\H) \nabla_E g(\Cm), \label{eq:Riemannian_gradient}
\end{align}
where $\nabla_E g(\Cm)$ is the Euclidean gradient of $g(\Cm)$, which is given by $\Big[\frac{\partial g(\Cm)}{\partial \cv_{11}} \ \dots \ \frac{\partial g(\Cm)}{\partial \cv_{1|\Xc_1|}} \ \frac{\partial g(\Cm)}{\partial \cv_{21}} \ \dots \ \frac{\partial g(\Cm)}{\partial \cv_{2|\Xc_2|}}\Big]$ with
\begin{align}
\frac{\partial g(\Cm)}{\partial \cv_{kn}} &= - \Bigg(\sum_{\Xm \ne {\Xm'} \in \Xc} \exp \Big(\!-\frac{f(\Xm, {\Xm'})}{\epsilon}\Big) \Bigg)^{-1} \notag \\
&~ \times \! \sum_{\substack{\Xm = \sqrt{PT}[\cv_{1i}\ \cv_{2l}] \\ \ne {\Xm'} = \sqrt{PT}[\cv_{1j}\ \cv_{2m}], \\ kn \in \{1i,1j,2l,2m\}}} \!\!\exp \Big(\!-\frac{f(\Xm, {\Xm'})}{\epsilon}\Big) \frac{\partial f(\Xm,{\Xm'})}{\partial \cv_{kn}}.
\end{align}
In our proposed criteria, $f(\Xm,{\Xm'})$ is given by
\begin{equation}
\!\!f(\Xm, {\Xm'}) \!\defeq\! \begin{cases}
\frac{1}{N} \meanLLR,\! &\text{for the criterion~\eqref{eq:criterion_minMean},} \\
\trace\big((\Id_T\!+\!{\Xm'}{\Xm'}^\H)^{-1}\!\Xm \Xm^\H\big),\! &\text{for the criterion~\eqref{eq:criterion_minTrace}}. 
\end{cases} 
\label{eq:f(X,hatX)}
\end{equation}
Essentially, we would like to compute the derivative of $d(\Xm \to {\Xm'}) =  \trace[\big(\Id_T+{\Xm'}{\Xm'}^\H\big)^{-1}\Xm \Xm^\H]$ (the derivative of $\trace[\big(\Id_T+{\Xm'}{\Xm'}^\H\big)^{-1}]$ is similar) and $\psi(\Xm, {\Xm'}) \defeq \log\frac{\det\big( \Id_T+ {\Xm'}
	{\Xm'}^\H \big)}{\det\left( \Id_T+ \Xm \Xm^\H
	\right)}$.
 With $\Xm = \sqrt{PT}[\cv_{1i}\ \cv_{2l}]$ and ${\Xm'} = \sqrt{PT}[\cv_{1j}\ \cv_{2m}]$, after some manipulations, we have that:
	\begin{align}
	\frac{\partial d(\Xm \!\to\! {\Xm'})}{\partial \cv_{1n}} &= \begin{cases}
	2\Big(\frac{1}{PT} \Id_T \!+\! \cv_{1j}\cv_{1j}^\H \!+\! \cv_{1n}\cv_{1n}^\H\Big)^{-1} \cv_{1n}, ~\text{if~} n = i, \\
	2\Big(\frac{\cv_{1n}^\H \Am^{-1}(\cv_{1i}\cv_{1i}^\H + \cv_{2l}\cv_{2l}^\H)\Am^{-1}\cv_{1n}(\Id_T+\Am^{-1})}{(1+\cv_{1n}^\H\Am^{-1}\cv_{1n})^2} \Big.\\
	\quad \Big.- \frac{\Am^{-1}(\cv_{1i}\cv_{1i}^\H + \cv_{2l}\cv_{2l}^\H)\Am^{-1}}{1+\cv_{1n}^\H\Am^{-1}\cv_{1n}}\Big) \cv_{1n}, \\ 
	\qquad\text{~with~} \Am \defeq \frac{1}{PT} \Id_T + \cv_{2m}\cv_{2m}^\H, ~\text{if~} n = j.
	\end{cases} \\ 
	\frac{\partial d(\Xm \!\to\!{\Xm'})}{\partial \cv_{2n}} &= \begin{cases}
	2\Big(\frac{1}{PT} \Id_T \!+\! \cv_{2m}\cv_{2m}^\H \!+\! \cv_{2n}\cv_{2n}^\H\Big)^{-1} \cv_{2n}, ~\text{if~} n = l, \\
	2\Big(\frac{\cv_{2n}^\H \Bm^{-1}(\cv_{2l}\cv_{2l}^\H + \cv_{1i}\cv_{1i}^\H)\Bm^{-1}\cv_{2n}(\Id_T+\Bm^{-1})}{(1+\cv_{2n}^\H\Bm^{-1}\cv_{2n})^2} \Big. \\  \quad \Big. -\frac{\Bm^{-1}(\cv_{2l}\cv_{2l}^\H + \cv_{1i}\cv_{1i}^\H)\Bm^{-1}}{1+\cv_{2n}^\H\Bm^{-1}\cv_{2n}}\Big) \cv_{2n}, \\ \qquad\text{~with~} \Bm \defeq \frac{1}{PT} \Id_T + \cv_{1j}\cv_{1j}^\H, ~\text{if~} n = m.
	\end{cases}
	\end{align}
	and $\psi(\Xm, {\Xm'}) = \log \frac{Q-|\cv_{1i}^\H\cv_{2l}|^2}{Q-|\cv_{1j}^\H\cv_{2m}|^2}$ with $Q\defeq \big(1+\frac{1}{PT}\big)^2$,
	\begin{align}
	\frac{\partial \psi(\Xm,{\Xm'})}{\partial \cv_{1n}} &= \begin{cases}
	2 \frac{\cv_{2l}\cv_{2l}^\H}{Q-|\cv_{1n}^\H\cv_{2l}|^2} \cv_{1n}, & \text{if~} n = i, \\
	2 \frac{\cv_{2m}\cv_{2m}^\H}{Q-|\cv_{1n}^\H\cv_{2m}|^2} \cv_{1n}, & \text{if~} n = j.
	\end{cases} \\ 
	\frac{\partial \psi(\Xm,{\Xm'})}{\partial \cv_{2n}} &= \begin{cases}
	2 \frac{\cv_{1i}\cv_{1i}^\H}{Q-|\cv_{1i}^\H\cv_{2n}|^2} \cv_{2n}, & \text{if~} n = l, \\
	2 \frac{\cv_{1j}\cv_{1j}^\H}{Q-|\cv_{1j}^\H\cv_{2n}|^2} \cv_{2n}, & \text{if~} n = m.
	\end{cases} 
	\end{align}
	
For $\Xcal_1 = \arg\displaystyle\max_{\Xcal_1} d_{12}(\Xcal)$ and $\Xcal_2 = \arg\displaystyle\max_{\Xcal_2} d_{21}(\Xcal)$ in the alternating optimization, the gradients are computed in a similar way to $d(\Xm \!\to\! {\Xm'})$.

With the gradient computed, we employ the manopt toolbox~\cite{manopt} to solve the optimization with gradient descent on the manifold.
\end{document}